\documentclass[12pt]{article}
\usepackage{amsmath}
\usepackage{amssymb}
\usepackage{amsfonts}
\usepackage{latexsym}
\usepackage{color}
\usepackage{graphicx}

\catcode `\@=11 \@addtoreset{equation}{section}

\catcode `\@=12

\newcommand{\be}{\begin{equation}}
\newcommand{\en}{\end{equation}}
\newcommand{\bea}{\begin{eqnarray}}
\newcommand{\ena}{\end{eqnarray}}
\newcommand{\beano}{\begin{eqnarray*}}
\newcommand{\enano}{\end{eqnarray*}}
\newcommand{\bee}{\begin{enumerate}}
\newcommand{\ene}{\end{enumerate}}

\newcommand{\Hil}{{\cal H}}

\newcommand{\I}{{\cal I}}

\newcommand{\F}{{\cal F}}

\newcommand{\E}{{\cal E}}

\newcommand{\1}{1 \!\!\! 1}
\newtheorem{thm}{Theorem}
\newtheorem{cor}[thm]{Corollary}
\newtheorem{lemma}[thm]{Lemma}
\newtheorem{prop}[thm]{Proposition}

\newenvironment{proof}{\noindent {\bf Proof:}}{\hfill$\Box$}

\begin{document}

\thispagestyle{empty}

\vspace*{1cm}

\begin{center}
{\Large \bf Mathematical aspects of intertwining operators: the role of Riesz bases}   \vspace{2cm}\\

{\large F. Bagarello}\\
  Dipartimento di Metodi e Modelli Matematici,
Facolt\`a di Ingegneria, Universit\`a di Palermo, I-90128  Palermo, Italy\\
e-mail: bagarell@unipa.it
\end{center}


\vspace*{2cm}

\begin{abstract}
\noindent In this paper we continue our analysis of intertwining
relations for both self-adjoint and not self-adjoint operators. In
particular, in this last situation, we discuss the connection with
pseudo-hermitian quantum mechanics and the role of Riesz bases.

\end{abstract}

\vspace{2cm}

\vfill

\newpage

\section{Introduction}

In a series of recent papers, \cite{bag1,bag2,bag3,bag4}, we have
proposed a new technique which produces, given two operators $h_1$
and $x$, a hamiltonian $h_2$ which has (almost) the same spectrum of
 $h_1$ and whose respective eigenstates are
related by the intertwining operator (IO) $x$. More precisely,
calling $\sigma(h_j)$, $j=1,2$, the set of eigenvalues of $h_j$, we
find that $\sigma(h_2)\subseteq\sigma(h_1)$. These results extend
what was discussed in the previous literature on this subject,
\cite{intop}, and have the advantage of proposing a constructive
procedure: while in \cite{intop} the existence of $h_1, h_2$ and of
an operator $x$ satisfying the  intertwining condition $h_1x=xh_2$
is assumed, in \cite{bag1}-\cite{bag4} we explicitly construct $h_2$
from $h_1$ and $x$ in such a way that $h_2$ satisfies a {\em weak
form} of $h_1x=xh_2$. Moreover, as mentioned above,
$\sigma(h_2)\subseteq\sigma(h_1)$ and the eigenvectors are related
in a {\em standard} way: if
$h_1\varphi_n^{(1)}=\epsilon_n\varphi_n^{(1)}$ then, if
$x^\dagger\varphi_n^{(1)}\neq0$,
$h_2\left(x^\dagger\varphi_n^{(1)}\right)=\epsilon_n\left(x^\dagger\varphi_n^{(1)}\right)$,
see \cite{bag1}-\cite{bag4}. It is well known that this procedure is
strongly related to the supersymmetric quantum mechanics, see
\cite{CKS} and \cite{jun} for two rather complete overviews.

In \cite{bag5,bag6,bag7} we have also discussed some relations
between IO and the so-called {\em pseudo-hermitian quantum
mechanics}, \cite{mosta2,mosta3} and \cite{mosta} for a review, in connection with the
so-called pseudo-bosons, which are excitations arising from a
deformation of the canonical commutation relation. In particular, in \cite{bag5}
the role of Riesz bases appeared clearly, and the
operators intertwined by $x$ were not, in general, self-adjoint.
This has suggested an extension of our previous results to the
situation in which the { hamiltonian $h_1$ is not self-adjoint} but
is rather, for instance, pseudo-hermitian.  This is the contain of Section III,
which follows a section dedicated to some mathematical aspects of
the self-adjoint situation. Our conclusions are contained in Section
IV. It should be mentioned that the analysis of IO for non self-adjoint operators was already considered, more on a physical side, in \cite{aicd}, and more recently in \cite{bdgl}. In none of these papers, however, the role of Riesz bases was considered.

\section{Some mathematical aspects of the IOs}

Let $h_1$ be a self-adjoint hamiltonian on the Hilbert space $\Hil$,
$h_1=h_1^\dagger$, whose normalized eigenvectors, $\varphi_n^{(1)}$,
satisfy the following equation:
$h_1\varphi_n^{(1)}=\epsilon_n\varphi_n^{(1)}$,
$n\in\I_1\subseteq\Bbb{N}_0=\Bbb{N}\cup\{0\}$. Let us now consider
an operator $x$ on $\Hil$ such that, calling $N_1:=x\,x^\dagger$ and
$N_2:=x^\dagger\, x$, the following commutation rule (to be
considered in the sense of unbounded operators, in general) is
satisfied: $[N_1,h_1]=0$. Both $N_1$ and $N_2$ are positive
operators, but they could have zero in their spectra. If this is the
case, then the related $N_j$ is not invertible. Since $h_1$ and
$N_1$ commute, they can be diagonalized simultaneously. Hence it is
natural to assume that the $\varphi_n^{(1)}$'s are also eigenstates
of $N_1$. Summarizing we have \be
h_1\varphi_n^{(1)}=\epsilon_n\varphi_n^{(1)}, \qquad
N_1\varphi_n^{(1)}=\nu_n\varphi_n^{(1)}, \label{21}\en for all
$n\in\I_1$. We call $\F_1=\{\varphi_n^{(1)}, \,n\in\I_1\}$ the set
of these states. A second natural working assumption concerns the
nature of $\F_1$ which is assumed here to be complete in $\Hil$ and
orthonormal: \be
\left<\varphi_n^{(1)},\varphi_m^{(1)}\right>=\delta_{n,m},\qquad
\sum_{n\in\I_1}P_n^{(1)}=\1, \label{22}\en where $\1$ is the
identity in $\Hil$ and we have introduced the following operators:
\be
P_{n,m}^{(1)}f=\left<\varphi_n^{(1)},f\right>\,\varphi_m^{(1)},\qquad
\mbox{and } P_n^{(1)}:=P_{n,n}^{(1)}, \label{23}\en for all
$f\in\Hil$. It is well known that the $P_n^{(1)}$'s are orthogonal
projectors: $P_n^{(1)}P_m^{(1)}=\delta_{n,m}\,P_n^{(1)}$ and $\left(P_n^{(1)}\right)^\dagger=P_n^{(1)}$. It is also
known that the orthogonality of the vectors in $\F_1$ is automatic
if the eigenvalues of $h_1$ are all different:
$\epsilon_n\neq\epsilon_m$, $\forall n\neq m$. Under our
assumptions we can write $h_1$ and $N_1$ as \be
h_1=\sum_{n\in\I_1}\,\epsilon_n\,P_n^{(1)}, \qquad
N_1=\sum_{n\in\I_1}\,\nu_n\,P_n^{(1)}. \label{24}\en Moreover, since
$N_1\geq0$, all its eigenvalues are non negative: $\nu_n\geq0$,
$\forall n\in\I_1$. Of course $N_2\geq0$ as well, and since they are
related by the commutator $[x,x^\dagger]$, $N_1=N_2+[x,x^\dagger]$,
we see that,
if $[x,x^\dagger]>0$, then $N_1>0$. If, on the contrary,
$[x,x^\dagger]<0$, then $N_2>0$.

Hence the invertibility of $N_1$ or $N_2$ can be established if
$[x,x^\dagger]$ is strictly positive or negative defined. For
instance, if $x=a$, where $[a,a^\dagger]=\1$, it is clear that
$N_1=a\,a^\dagger>0$.

\vspace{2mm}

{\bf Remark:--} This is not the only case in which we can deduce the
existence of $N_j^{-1}$. For instance, if $x=a$ with $a$ satisfying
the modified commutation relation
$[a,a^\dagger]_q:=a\,a^\dagger-q\,a^\dagger\,a=\1$, $q\in[0,1]$, then, since
$N_1=[a,a^\dagger]_q+q\,N_2$, we find that $N_1>0$.

\vspace{2mm}

Let us now define the following vectors: \be
\varphi_n^{(2)}:=x^\dagger\,\varphi_n^{(1)}, \label{25}\en for
$n\in\I_1$. It may happen that for some $n$ in $\I_1$ the action of
$x^\dagger$ on $\varphi_n^{(1)}$ returns the zero vector. This means
that $\ker(x^\dagger)$ is non trivial. Of course, if
$\varphi_{n_0}^{(1)}\in\ker(x^\dagger)$, then
$N_1\varphi_{n_0}^{(1)}=0$ so that $\nu_{n_0}=0$ and the operator
$N_1$ is not invertible. Viceversa, if $N_1$ is not invertible, then $\ker(N_1)$ contains some non zero vectors. Let $\Psi$ be such a vector. Then $N_1\Psi=0$ and, consequently, $\|x^\dagger\Psi\|^2=0$. Hence $x^\dagger\Psi=0$, which means that $\Psi\in\ker(x^\dagger)$. In other words: $\Psi\in\ker(N_1)$ if and only if $\Psi\in\ker(x^\dagger)$.

Let now $\I_2=\{n\in\I_1:\,
\varphi_n^{(2)}\neq0\}$, and let
$\F_2=\{\varphi_n^{(2)},\,n\in\I_2\}$. Of course, if
$\ker(x^\dagger)=\{0\}$, then $\I_1=\I_2$, but in general we only
have the inclusion $\I_2\subseteq\I_1$.

The set $\F_2$ could consist, at least in principle, of very few
vectors if compared with $\F_1$. So the problem of completeness of
$\F_2$ in $\Hil$ arises. Moreover, if $x$ is not unitary (or at least
isometric), we don't know if different $\varphi_n^{(2)}$'s are
orthogonal to each other or if they are, by chance, eigenstates of
some interesting operator. This is exactly what happens. In
particular we can prove the following
\begin{prop}\label{prop2}
Under the above hypotheses the $\varphi_n^{(2)}$'s satisfy the
eigenvalue equation \be N_2\varphi_n^{(2)}=\nu_n\varphi_n^{(2)},
\label{26}\en for all $n\in\I_2$. Moreover, if for $n,m\in\I_2$,
$n\neq m$, $\nu_n\neq\nu_m$, then
$\left<\varphi_n^{(2)},\varphi_m^{(2)}\right>=0$. Finally, the set
$\F_2$ is complete in $\Hil$ if and only if $N_2$ is invertible.

\end{prop}

\begin{proof}
Since $n\in\I_2$ the vector $\varphi_n^{(1)}$ does not belong to
$\ker(x^\dagger)$, so that $\varphi_n^{(2)}\neq0$. We have
$$N_2\varphi_n^{(2)}=\left(x^\dagger\,x\right)\left(\,x^\dagger\,\varphi_n^{(1)}\right)=x^\dagger\,
N_1\,\varphi_n^{(1)}=\nu_n\,\left(\,x^\dagger\,\varphi_n^{(1)}\right)=\nu_n\,\varphi_n^{(2)}$$
Then our second claim is straightforward.

Let us now assume that $N_2^{-1}$ exists. Then $\F_2$ is complete.
Indeed, let $f$ be an element of $\Hil$ such that
$\left<f,\varphi_n^{(2)}\right>=0$ for all $n\in\I_2$. Hence, for
these values of $n$, we also have
$\left<x\,f,\varphi_n^{(1)}\right>=0$.

We consider separately two cases: $\I_2=\I_1$ and $\I_2\subset\I_1$.
In the first case, since $\F_1$ is complete by assumption, we
conclude that $x\,f=0$, which also implies that $x^\dagger
x\,f=N_2\,f=0$. But, since $N_2$ is invertible by assumption, $f=0$.
Hence $\F_2$ is complete.

Suppose now that $\I_2\subset\I_1$. Then equality $\left<x\,f,\varphi_n^{(1)}\right>=0$ for all $n\in\I_2$, implies that  the vector $x\,f$ can be
written as a linear combination (which could involve infinite
elements but which is clearly  convergent) of the form
$x\,f=\sum_{k\in\Gamma}\alpha_k\varphi_k^{(1)}$, where
$\Gamma=\I_1\setminus\I_2$ and the $\alpha_k$'s are complex
constants. Taking the scalar product of both sides of this expansion
for $\varphi_l^{(1)}$, $l\in\Gamma$, since for these values of $l$
$\alpha_l=\left<\varphi_l^{(1)},x\,f\right>=\left<x^\dagger\,\varphi_l^{(1)},f\right>=0$,
we deduce that $\alpha_l=0$ for all $l\in\Gamma$. Hence $x\,f=0$
and, as before, since $N_2^{-1}$ does exist, $f=0$.

To prove the inverse implication we will show that if $N_2^{-1}$
does not exist, then $\F_2$ is not complete. Indeed, since
$N_2^{-1}$ does not exist there exists a vector $g\in\Hil$, $g\neq
0$, such that $N_2\,g=0$. This implies that
$0=\left<g,N_2\,g\right>=\|x\,g\|^2$, so that $x\,g=0$. Suppose now
that $\F_2$ is complete. Using the equality $x\,g=0$ we deduce that,
for all $n\in\I_2$,
$\left<g,\varphi_n^{(2)}\right>=\left<x\,g,\varphi_n^{(1)}\right>=0$.
Hence, due to our assumption on $\F_2$, we would have $g=0$ which is
against our original hypothesis.

\end{proof}

It is interesting to notice that, while it may happen that
$x^\dagger\varphi_n^{(1)}=0$ for some $n\in\I_1$, it never happens
that $x\,\varphi_n^{(2)}=0$ for any $n\in\I_2$. This fact, which is
clearly due to the procedure we are adopting, follows from the
following considerations:

first we remark that, since $\forall n\in\I_2$, $\varphi_n^{(2)}\neq 0$, then
$$
0\neq \|\varphi_n^{(2)}\|^2=\left<x^\dagger \varphi_n^{(1)},x^\dagger\varphi_n^{(1)}\right>=
\left< \varphi_n^{(1)},N_1\varphi_n^{(1)}\right>=\nu_n.
$$
Hence $\nu_n>0$ $\forall n\in\I_2$. But, for these $n$'s, we also have
$
x\,\varphi_n^{(2)}=x\,x^\dagger\,\varphi_n^{(1)}=N_1\,\varphi_n^{(1)}=\nu_n\,\varphi_n^{(1)}.
$
Then, as stated,  $x\,\varphi_n^{(2)}\neq0$ and moreover \be
\varphi_n^{(1)}=\frac{1}{\nu_n}\,x\,\varphi_n^{(2)}=\frac{1}{\|\varphi_n^{(2)}\|^2}\,x\,\varphi_n^{(2)}.
\label{28}\en The
normalized vectors associated to $\F_2$ are therefore
$\hat\F_2=\{\hat\varphi_n^{(2)}=\frac{1}{\sqrt{\nu_n}}\,\varphi_n^{(2)},\,n\in\I_2\}$.
By means of the Proposition 1 above we deduce that each
$\hat\varphi_n^{(2)}$ is an eigenstate of $N_2$, that they form an
orthonormal set, at least if all the $\nu_n$'s are different, and that $\hat\F_2$ is complete if and only if $N_2$
is invertible. In analogy with what we have done before, let us
define the operators $P_n^{(2)}$ and $\hat P_n^{(2)}$ as follows:
$P_n^{(2)}\,f=\left<\varphi_n^{(2)},f\right>\,\varphi_n^{(2)}$ and
$\hat
P_n^{(2)}\,f=\left<\hat\varphi_n^{(2)},f\right>\,\hat\varphi_n^{(2)}$.
It may be worth remarking that while the $\hat P_n^{(2)}$'s are
orthogonal projections, the $P_n^{(2)}$'s are not, since
$\left(P_n^{(2)}\right)^2\neq P_n^{(2)}$. It is also possible to prove the
following
\begin{cor}
Let us assume that for $n,m\in\I_2$, $n\neq m$, $\nu_n\neq\nu_m$, and
that $N_2^{-1}$ exists. Hence \be
N_2=\sum_{k\in\I_2}P_n^{(2)}=\sum_{k\in\I_2}\nu_n\,\hat P_n^{(2)},
\label{29}\en
\end{cor}
The proof, which makes use of the resolution of the identity for
$\hat \F_2$, is trivial and will not be given here. More results on
$\ker(x^\dagger)$ will be given, under generalized assumptions, in
the next section.

\vspace{2mm}

As in \cite{bag1}-\cite{bag4} we now define $h_2=N_2^{-1}(x^\dagger h_1 x)$. We know that $h_2\varphi_n^{(2)}=\epsilon_n\varphi_n^{(2)}$, for all $n\in\I_2$. Moreover, among other properties, we also know that $h_2=h_2^\dagger$ if and only if $h_1=h_1^\dagger$. Using Proposition \ref{prop2}, which applies since we are here assuming that $N_2^{-1}$ does exist, we deduce that $\F_2$ is complete so that $h_2$ can be written as
\be
h_2=\sum_{n=0}^\infty\epsilon_n\hat P_n^{(2)}=\sum_{n=0}^\infty\frac{\epsilon_n}{\nu_n}P_n^{(2)}.
\label{210}\en
Adopting the Dirac bra-ket notation we find that $xP_n^{(2)}=\nu_n|\varphi_n^{(1)}\left>\right<\varphi_n^{(2)}|$ and $P_n^{(1)}x=|\varphi_n^{(1)}\left>\right<\varphi_n^{(2)}|$ for all $n$. Then it follows that $xh_2=h_1x$, which means that $x$ intertwines between $h_2$ and $h_1$, as in the standard papers on this subject. We will recover this intertwining relation in the following section, generalizing our previous results, \cite{bag1}-\cite{bag4}.

\vspace{3mm}

{\bf Example:--} ({\em the ubiquitous harmonic oscillator})

Many aspects of our procedure can be quite well illustrated by means
of the canonical commutation relation $[a,a^\dagger]=\1$ arising
from the hamiltonian of a quantum harmonic oscillator, which in
suitable units and putting to zero the ground state energy is
$h_1=a^\dagger\,a$. If $\varphi_0^{(1)}$ is such that
$a\varphi_0^{(1)}=0$, then the eigenstates of $h_1$ are the usual
ones:
$\varphi_n^{(1)}=\frac{1}{\sqrt{n!}}\,(a^\dagger)^n\,\varphi_0^{(1)}$,
whose related set $\F_1$ is orthonormal and complete in $\Hil$. Let us now
take $x=a^\dagger$. Hence $N_1=x\,x^\dagger=h_1$ and
$N_2=x^\dagger\,x=a\,a^\dagger$. Hence $N_2$ is invertible and $[h_1,N_1]=0$.
Moreover we see that $\I_1=0,1,2,\ldots$ while, since
$x^\dagger\varphi_0^{(1)}=a\varphi_0^{(1)}=0$, $\I_2=1,2,\ldots$.
Hence $\I_2\subset\I_1$. Nevertheless the set $\hat\F_2$ coincides
exactly with $\F_1$. Hence it is complete in $\Hil$, as expected
because of Proposition 1. In this case we find easily that $h_2=N_2^{-1}(a^\dagger h_1 a)=aa^\dagger=h_1+\1$.

If on the other way we take $x=a$ then
$N_2=x^\dagger\,x=a^\dagger\,a=h_1$ which is not invertible.
This is in agreement with the fact that the set $\hat\F_2$ is now a
proper subset of $\F_1$ since $\varphi_0^{(1)}$ belongs to $\F_1$
but not to $\hat\F_2$. Hence $\hat\F_2$ is not complete in $\Hil$,
as expected because of Proposition 1. Moreover we find $h_2=h_1-\1$, whose eigenvalues are non negative since its eigenvectors are  $\{\varphi_1^{(1)}, \varphi_2^{(1)}, \varphi_3^{(1)},\ldots\}$.

\vspace{3mm}

{\bf Example:--} ({\em the deformed harmonic oscillator: quons})

Following \cite{bag4} we consider two operators, $B$ and $B^\dagger$, which satisfy the modified commutation relation $[B,B^\dagger]_q:=B,B^\dagger-q B^\dagger B=\1$, $q\in[0,1]$. Let $\varphi_0^{(1)}$ be the vacuum of $B$: $B\varphi_0^{(1)}=0$. Let furthermore $h_1=B^\dagger B$. Then, putting
\be\varphi_n^{(1)}=\frac{1}{\beta_0\cdots\beta_{n-1}}\,{B^\dagger}^n\,\varphi_0^{(1)}=
\frac{1}{\beta_{n-1}}B^\dagger\varphi_{n-1}^{(1)},\qquad n\geq 1,\label{21ex}\en
we have $h_1\varphi_n^{(1)}=\epsilon_n\varphi_n^{(1)}$, with $\epsilon_0=0$, $\epsilon_1=1$ and $\epsilon_n=1+q+\cdots+q^{n-1}$ for $n\geq 1$. Also, the normalization is found to be $\beta_n^2=1+q+\cdots+q^n$, for all $n\geq0$. Hence $\epsilon_n=\beta_{n-1}^2$ for all $n\geq1$.  The set of the $\varphi_n^{(1)}$'s  spans the Hilbert space $\Hil$ and they are mutually orhonormal: $<\varphi_{n}^{(1)},\varphi_{k}^{(1)}>=\delta_{n,k}$.

We now take, as in the previous example, $x=B^\dagger$. Then  $N_1=B^\dagger B$ and $N_2=BB^\dagger$ and, obviously, $[h_1,N_1]=0$. Moreover, since $N_2=BB^\dagger=\1+q\,B^\dagger B$, and since $B^\dagger B$ is a positive operator, $N_2^{-1}$ exists. We easily find that $h_2:=N_2^{-1}\left(x^\dagger\,h_1\,x\right)=q\,h_1+\1$
while
$$
\varphi_n^{(2)}=B\varphi_n^{(1)}= \left\{
\begin{array}{ll}
0\hspace{2.3cm} \mbox{ if } n=0  \\
\beta_{n-1}\varphi_{n-1}^{(1)}\qquad \mbox{ if } n\geq1.\\
\end{array}
\right.
$$
Then $\hat\F_2$ coincides
again with $\F_1$, and so it is complete in $\Hil$, as expected
because of Proposition 1. If we rather take $x=B$ it is not hard to check that completeness is lost because $
\varphi_0^{(1)}$ does not belong to $\hat\F_2$. This is in agreement with the fact that, since $N_2\varphi_0^{(1)}=0$, $N_2$ is not invertible. Incidentally we observe that $h_2=\frac{1}{q}(h_1-\1)$.

\section{Losing  self-adjointness}

In the previous section we have considered the case in which the two
operators $h_1$ and $h_2$ related by the intertwining operator $x$
are self-adjoint. Now we will remove this assumption and we will
discuss some interesting consequences of this more general
situation. In particular we will see that, in this new context, there are strong indications which suggest to replace o.n.
bases  by Riesz bases.

Let $\Theta_1$ be a non necessarily self-adjoint operator on $\Hil$
which admits a set $\F_1=\{\varphi_n^{(1)}, \,n\geq 0\}$ of
eigenstates: \be \Theta_1\varphi_n=\epsilon_n\varphi_n,\qquad n\geq
0, \label{31}\en for some (in general complex) $\epsilon_n$. In this
section we will always work under the simplifying
assumption that all these eigenvalues have multiplicity 1. This is
useful to simplify the formulation of our results, but it could be
avoided most of the times. However, examples of this situation are discussed, for instance, in \cite{bag5,bag6} and references therein. A class of new examples generalizing the so-called Landau levels and in which the multiplicity of each energetic level is infinity will be discussed in a paper which is now in preparation, \cite{abg}.  As in the previous section, we will
assume now that an operator $x$ exists, acting on $\Hil$, such that,
calling $N_1=x\,x^\dagger$ and $N_2=x^\dagger\,x$, $N_1$ commutes
(in the sense of unbounded operators, if needed) with $\Theta_1$ and
that $N_2$ is invertible. Depending on the fact that $x$ is
invertible by itself or not, we introduce two apparently different
operators: \be \Theta_2=\left\{
    \begin{array}{ll}
        x^{-1}\Theta_1\,x,\hspace{1.6cm}\mbox{ if } x^{-1} \mbox{ exists};  \\
        N_2^{-1}(x^\dagger\Theta_1 x), \hspace{0.8cm} \mbox{ otherwise}.\\
       \end{array}
        \right.
\label{32}\en To distinguish between these two we will sometimes call in the
following $\Theta_2^{(\alpha)}=x^{-1}\Theta_1\,x$ and
$\Theta_2^{(\beta)}=N_2^{-1}(x^\dagger\Theta_1 x)$. It is clear
that, when $x^{-1}$ exists, $\Theta_2^{(\beta)}$  coincides with
$\Theta_2^{(\alpha)}$. However, when $x^{-1}$ does not exist, $\Theta_2^{(\alpha)}$ makes no sense but we can
still introduce $\Theta_2^{(\beta)}$. When our statements apply  both for $\Theta_2^{(\alpha)}$ and for $\Theta_2^{(\beta)}$, to simplify the notation we just use $\Theta_2$.

\vspace{2mm}

{\bf Remark:--} Analogously to what discussed in the previous section, the existence of an operator $x$ satisfying $[xx^\dagger,\Theta_1]=0$ has interesting consequences concerning the possibility of finding
a second operator $\Theta_2$ with (almost) the same eigenvalues (real or complex, now it doesn't matter) as $\Theta_1$, see below.

\vspace{2mm}

We define, as usual, \be \varphi_n^{(2)}=x^\dagger\varphi_n^{(1)},
\qquad n\geq 0, \label{33}\en which, if
$\varphi_n^{(1)}\notin\ker(x^\dagger)$, is an eigenstate of
$\Theta_2$ with eigenvalue $\epsilon_n$, independently of whether
$\Theta_1$ is self-adjoint or not. Using $[N_1,\Theta_1]=0$ it is in
fact quite easy to check that, if $\varphi_n^{(2)}\neq0$,
$\Theta_2\varphi_n^{(2)}=\epsilon_n\varphi_n^{(2)}$, for all
$n\geq0$.

\vspace{2mm}

{\bf Remark:--} It should be mentioned that, if $x$ is invertible,
we could also define the eigenstates of $\Theta_2^{(\alpha)}$ as
$\tilde\varphi_n^{(2)}=x^{-1}\varphi_n^{(1)}$, which seems more
appropriate since nothing should be required to the commutator
between $N_1$ and $\Theta_1$, at least as far as the eigenvalue
equation for $\Theta_2^{(\alpha)}$ is concerned. But, in our
previous papers, \cite{bag1}-\cite{bag4}, we have focused our
attention on the situation in which $x^{-1}$ does not necessarily
exist, and this will be our main interest also here. We
will return to this aspect later on.

\vspace{2mm}

Let us now go back for a moment to the requirement
$\varphi_n^{(1)}\notin\ker(x^\dagger)$ above. It is possible to
prove the following result, which extends what already stated in Section II.

\begin{lemma}\label{lemma1}
With the above definitions, for a given $n\geq 0$,
$\varphi_n^{(1)}\in\ker(x^\dagger)$ if and only if
$\varphi_n^{(1)}\in\ker(N_1)$ or, equivalently, if and only if
$\varphi_n^{(2)}\in\ker(x)$.
\end{lemma}
\begin{proof}
We only prove here that if $\varphi_n^{(2)}\in\ker(x)$ then
$\varphi_n^{(1)}\in\ker(x^\dagger)$. Indeed our assumption implies
that $0=x\varphi_n^{(2)}=x\,x^\dagger\varphi_n^{(1)}$, so that
$\varphi_n^{(1)}\in\ker(N_1)$ which in turns implies, using the first statement of this Lemma, that $\varphi_n^{(1)}\in\ker(x^\dagger)$.

\end{proof}

From the definition of $\Theta_2^{(\alpha)}$ it is clear that $x$ is
an intertwining operator, since $x\Theta_2^{(\alpha)}=\Theta_1\,x$.
What is not evident is whether $x\Theta_2^{(\beta)}=\Theta_1\,x$ is
also true. We have already considered this problem in the previous section. It is not hard to see that the answer is affirmative also in the present situation.
This is a consequence of the fact that
$[x^\dagger\,\Theta_1\,x,N_2]=[x^\dagger\,\Theta_1\,x,N_2^{-1}]=0$,
which can be proved easily. A detailed analysis produces, other than
these, the following commutation rules
$$
[\Theta_2^{(j)},N_2]=[\Theta_2^{(j)},N_2^{-1}]=[(\Theta_2^{(j)})^\dagger,N_2]=[(\Theta_2^{(j)})^\dagger,N_2^{-1}]=0,
$$
as well as $$ [\Theta_1^\dagger,N_1]=[x^\dagger\Theta_1^\dagger
x,N_2]=[x^\dagger\Theta_1^\dagger
x,N_2^{-1}]=[x(\Theta_2^{(j)})^\dagger x^\dagger,N_1]=0, \quad
j=\alpha,\beta,
$$ and the following intertwining relations, all arising from our
assumptions and from (\ref{31}) and (\ref{32}):
\be
x\Theta_2^{(j)}=\Theta_1\,x,\quad \Theta_2^{(j)}x^\dagger=x^\dagger \Theta_1,\quad
j=\alpha,\beta.
\label{34}
\en
Of course, the second equality in (\ref{34}) is just the adjoint of the first one only if $\Theta_1$ and $\Theta_2^{(j)}$ are self-adjoint, otherwise they are different.
An interesting consequence is deduced if $\Theta_2=\Theta_1^\dagger$, which is important, as we will discuss in the
following, in the context of pseudo-hermitian quantum mechanics (PHQM), \cite{mosta,mosta2,mosta3}.  In this case $x$ is not only an IO but
it also commutes with $\Theta_1+\Theta_1^\dagger$, as well as $x^\dagger$ does. Hence $N_1$, $N_2$ and $\Theta_1+\Theta_2$ are
three self-adjoint operators such that $[N_1,\Theta_1+\Theta_2]=[N_2,\Theta_1+\Theta_2]=0$, but, in general, $[N_1,N_2]\neq0$. So they are not expected to admit a set of common eigenvectors.

An evident difference between $N_1$ and $N_2$ is that, while $N_2$
is strictly positive by assumption, $N_1$ needs not to be
invertible. On the other hand our original assumption
$[N_1,\Theta_1]=0$ is reflected by $[N_2,\Theta_2]$, which is
also zero. However, analogously to what we observed in the previous
section, if $x$ is such that $[x,x^\dagger]\geq0$ (in the sense of
the operators), then $N_1=[x,x^\dagger]+N_2$ is also strictly positive so that it is invertible. If this is the case, the
commutation rules listed before can be enriched by other rules involving
$N_1^{-1}$, which will play no role here, and therefore will not be
considered. More interesting is the following

\begin{lemma}\label{lemma3}
If $N_2^{-1}$ exists, $[\Theta_1,N_1]=0$ and if (\ref{31}) holds,
then $\Theta_2=\Theta_2^{\dagger}$ if and only if
$\Theta_1=\Theta_1^{\dagger}$.
\end{lemma}
\begin{proof}
We give two different proofs for $\Theta_2^{(\alpha)}$ and $\Theta_2^{(\beta)}$.

Let us first suppose that
$\Theta_2^{(\alpha)}=(\Theta_2^{(\alpha)})^{\dagger}$. This means
that we are working under the assumption that $x^{-1}$ does exist.
Hence $x^{-1}\Theta_1\,x=x^\dagger\Theta_1^\dagger(x^{-1})^\dagger$,
which is equivalent to $\Theta_1=N_1\Theta_1^\dagger N_1^{-1}$
which, since $[N_1,\Theta_1^\dagger]=0$, implies that
$\Theta_2^{(\alpha)}=(\Theta_2^{(\alpha)})^{\dagger}$ if and only if
$\Theta_1=\Theta_1^\dagger$.

Let us now suppose  that
$\Theta_2^{(\beta)}=(\Theta_2^{(\beta)})^{\dagger}$ (we are no
longer requiring  $x$ to be invertible). This is equivalent to
$N_2^{-1}(x^\dagger\Theta_1 x)=(x^\dagger\Theta_1^\dagger
x)N_2^{-1}$, which, since $[x^\dagger\,\Theta_1\,x,N_2^{-1}]=0$, is
equivalent to $x^\dagger\Theta_1 x=x^\dagger\Theta_1^\dagger x$.
Now, using (\ref{34}), this can be rewritten as
$\Theta_2^{(\alpha)}x^\dagger
x=(\Theta_2^{(\alpha)})^{\dagger}x^\dagger x$ which can be
multiplied from the right by $N_2^{-1}$, giving back
$\Theta_2^{(\alpha)}=(\Theta_2^{(\alpha)})^{\dagger}$. Hence, $\Theta_2^{(\beta)}=(\Theta_2^{(\beta)})^{\dagger}$ if and only if $\Theta_2^{(\alpha)}=(\Theta_2^{(\alpha)})^{\dagger}$  which, in turns, is equivalent to $\Theta_1=\Theta_1^\dagger$.
\end{proof}

\begin{prop}
Under the  assumptions of Lemma \ref{lemma3}, let us suppose
that, for a fixed $k\geq0$, we have $\varphi_k^{(1)}\notin\ker(N_1)$.
Then $\varphi_k^{(1)}$ is also eigenstate of $N_1$ with a strictly
positive eigenvalue $\nu_k$. Moreover $\varphi_k^{(2)}$ is a
non-zero eigenstate of $N_2$ with eigenvalue $\nu_k$.

Furthermore, if $\nu_k$ has multiplicity one, $m(\nu_k)=1$,
then $\varphi_k^{(1)}$ is also eigenstate of $\Theta_1^\dagger$ with
eigenvalue $\overline{\epsilon_k}$, and $\varphi_k^{(2)}$ is an
 eigenstate of $\Theta_2^\dagger$ with eigenvalue $\overline{\epsilon_k}$.

Finally, if $m(\nu_k)=1$ for all $k\geq 0$, then $\F_1$ and
$\F_2$ are orthogonal systems in $\Hil$. Moreover:

$\bullet$ if $\F_1$ is complete in $\Hil$, then
$[\Theta_1,\Theta_1^\dagger]=0$;

$\bullet$ if $\F_2$ is complete in $\Hil$, then
$[\Theta_2,\Theta_2^\dagger]=0$.

\end{prop}

\begin{proof}
Since $[\Theta_1,N_1]=0$, and since $m(\epsilon_k)=1$, it follows
that $N_1\varphi_k^{(1)}$ must be proportional to $\varphi_k^{(1)}$
itself. Let $\nu_k$ be this proportionality constant. Hence
$N_1\varphi_k^{(1)}=\nu_k\varphi_k^{(1)}$ and
$\nu_k=\frac{\|x^\dagger\varphi_k^{(1)}\|^2}{\|\varphi_k^{(1)}\|^2}$,
which is strictly positive since
$\varphi_k^{(1)}\notin\ker(x^\dagger)$, see Lemma \ref{lemma1}. This
same Lemma also implies that $\varphi_k^{(2)}\neq 0$, and we have
$N_2\varphi_k^{(2)}=x^\dagger\,x\,x^\dagger\varphi_k^{(1)}=x^\dagger
N_1\varphi_k^{(1)}=\nu_k\varphi_k^{(2)}$.

Now, we notice that $\varphi_k^{(1)}\notin \ker(\Theta_1^\dagger)$
and that $\varphi_k^{(2)}\notin \ker(\Theta_2^\dagger)$. Indeed we
have
$$
\left<\varphi_k^{(1)},\Theta_1^\dagger\varphi_k^{(1)}\right>=\left<\Theta_1\varphi_k^{(1)},\varphi_k^{(1)}\right>=
\overline{\epsilon_k}\|\varphi_k^{(1)}\|^2,
$$
and analogously
$\left<\varphi_k^{(2)},\Theta_2^\dagger\varphi_k^{(2)}\right>=
\overline{\epsilon_k}\|\varphi_k^{(2)}\|^2$, which are both
different from zero.

If we now assume that $m(\nu_k)=1$, since
$[N_1,\Theta_1^\dagger]=0$, we conclude that
$\Theta_1^\dagger\varphi_k^{(1)}$ is proportional to
$\varphi_k^{(1)}$ itself, and the proportionality constant is easily
found to be $\overline{\epsilon_k}$:
$\Theta_1^\dagger\varphi_k^{(1)}=\overline{\epsilon_k}\varphi_k^{(1)}$.
In similar way we also deduce that
$\Theta_2^\dagger\varphi_k^{(2)}=\overline{\epsilon_k}\varphi_k^{(2)}$.

Finally, let us assume that $m(\nu_k)=1$ for all $k\geq0$. This
implies that, taking $j=1,2$, since different $\varphi_k^{(j)}$'s
are eigenvectors of self-adjoint operators $N_j$ corresponding to
different eigenvalues, they must be orthogonal:
$$
\left<\varphi_k^{(j)},\varphi_n^{(j)}\right>=0
$$
if $k\neq n$, for $j=1,2$. Moreover, if for instance $\F_1$ is
complete in $\Hil$, hence it is an o.n. basis. Therefore our last claim
easily follows from the fact that
$\Theta_1\Theta_1^\dagger\varphi_k^{(1)}=
\Theta_1^\dagger\Theta_1\varphi_k^{(1)}=|\epsilon_k|^2\varphi_k^{(1)}$,
for all $k\geq0$.

\end{proof}

\subsection{The role of Riesz bases}

Since $\Theta_k$, $k=1,2$, are not, in general, self-adjoint
operators, the sets of their eigenstates $\F_1$ and $\F_2$ are not
orthonormal, in general. This is one of the reasons why we are now
considering the role of Riesz bases in the present context. The
second reason, as already stated in the Introduction, is that in a series of
recent papers, \cite{bag5,bag6,bag7}, Riesz bases have already
appeared in analogous problems, and they have shown to be quite
relevant.

Let us now assume that the set $\F_1$ of eigenstates of
$\Theta_1$ is a Riesz basis for $\Hil$. This means that a bounded
operator $T$ exists, with bounded inverse $T^{-1}$, and an o.n.
basis $\E=\{e_n\in\Hil,\,n\geq0\}$, such that $\varphi_n^{(1)}=Te_n$
for all $n$. Equivalently, \cite{you,chri}, we can say that the
vectors of $\F_1$ are linearly independent and two constants exist,
$0<A\leq B<\infty$, such that, for all $f\in\Hil$,
$$
A\|f\|^2\leq \sum_{n\geq 0}\left|\left<\varphi_n^{(1)},f\right>\right|^2\leq B \|f\|^2.
$$
Then a bounded operator (the {\em frame operator}) $S_1:=\sum_{n\geq
0}|\varphi_n^{(1)}\left>\right<\varphi_n^{(1)}|=T\,T^\dagger$
exists, with bounded inverse, and the set
$\tilde\F_1=\{\tilde\varphi_n^{(1)}=S_1^{-1}\varphi_n^{(1)}\}$ is
biorthogonal to $\F_1$:
$\left<\varphi_n^{(1)},\tilde\varphi_k^{(1)}\right>=\delta_{n,k}$,
for all $n, k\geq0$. Moreover $\tilde\F_1$ is a Riesz basis by
itself, since $\tilde\varphi_n^{(1)}=\tilde Te_n$ for all $n$, with
$\tilde T=S_1^{-1} T$. Indeed, $\tilde T$ is bounded with bounded
inverse. Also, the following resolutions of the identity can be deduced: \be
\sum_{n\geq
0}|\tilde\varphi_n^{(1)}\left>\right<\varphi_n^{(1)}|=\sum_{n\geq
0}|\varphi_n^{(1)}\left>\right<\tilde\varphi_n^{(1)}|=\1
\label{35}\en Let now $\Theta_1$, $\F_1$ and $x$ be as in the first
part of Section III. It is interesting to analyze the nature of
$\F_2$. The first easy result is the following
\begin{lemma}\label{lemma2}
Let $\F_1$ be a Riesz basis. Then the following are equivalent: (a)
$x=S_1^{-1}$; (b) the set
$\F_2=\{\varphi_n^{(2)}=x^\dagger\varphi_n^{(1)},\, n\geq0\}$ is a
Riesz basis biorthogonal to $\F_1$.
\end{lemma}
\begin{proof}
The proof that (a) implies (b) follows from our previous discussion,
noticing that in this case $\F_2=\tilde\F_1$.

The converse implication can be proved as follows: first we notice
that
$$
\delta_{n,k}=\left<\varphi_n^{(1)},\varphi_k^{(2)}\right>=\left<Te_n,x^\dagger Te_k\right>=\left<T^\dagger x Te_n,e_k\right>,
$$
for all $n, k$, which implies that $T^\dagger x T=\1$. Then,
recalling that $T$ is invertible, we get
$x=(TT^\dagger)^{-1}=S_1^{-1}$.

\end{proof}

By means of the resolutions (\ref{35}) we can easily deduce that, if
$\F_2=\tilde\F_1$, \be \Theta_1=\sum_{n\geq
0}\epsilon_n|\varphi_n^{(1)}\left>\right<\varphi_n^{(2)}|,
\quad\mbox{ and } \Theta_2=\sum_{n\geq
0}\epsilon_n|\varphi_n^{(2)}\left>\right<\varphi_n^{(1)}|,
\label{36}\en from which it is possible to deduce that
$\Theta_1=\Theta_2^\dagger$ if and only if $\epsilon_n$ is real for
all $n$.

Of course, having two biorthogonal Riesz bases $(\F_1,\F_2)$   of
eigenstates of $(\Theta_1,\Theta_2)$ looks quite interesting and is a
natural extension of what happens for self-adjoint operators. However, the
following proposition can be seen as a sort of no-go result. Indeed
it states that, under the hypotheses we are considering here,
$\Theta_2$ coincides with $\Theta_1$ and $\F_2$ coincides with $\F_1$. In
other words, we are just only apparently introducing new vectors and
a new operator.

\begin{prop}\label{prop1}
Let $\F_1$ be a set of eigenstates of $\Theta_1$ which is also a
Riesz basis, and let us suppose that $x=S_1^{-1}$ and that
$[\Theta_1,N_1]=0$. Hence $[\Theta_1,x]=0$ and $\Theta_2=\Theta_1$.
Moreover, for all $n\geq 0$, $\varphi_n^{(2)}$ is proportional to
$\varphi_n^{(1)}$.
\end{prop}
\begin{proof}
First we remark that, since $x=S_1^{-1}$, and since $S_1$ is
self-adjoint, then $N_1=N_2=S_1^{-2}$, which is clearly invertible.
Moreover $\Theta_2^{(\alpha)}$ and $\Theta_2^{(\beta)}$ coincide,
since $x^{-1}$ exists and is equal to $S_1$. Now, since $N_1$ is a
positive and bounded operator commuting with $\Theta_1$, it is
known, \cite{rs}, that there exists an unique positive operator, the
square root of $N_1$, which commutes with all the operators which
commute with $N_1$. Of course this positive square root is
$S_1^{-1}$ itself, and then it follows that $[\Theta_1,S_1^{-1}]=0$.
Our claims now easily follow.

\end{proof}

This result suggests that the assumptions contained in Lemma
\ref{lemma2} are too restrictive and should be weakened. This is exactly what we will do in
the rest of this section. We begin with the following proposition,
related to the structure of $\Theta_1$ in connection with its
pseudo-hermiticity. This result generalizes those contained in \cite{mosta3}.

\begin{prop}
If $\Theta_1$ admits a basis of eigenvectors which is a Riesz basis,
with real eigenvalues, then there exists an operator $T$,
bounded with bounded inverse, such that $\Theta_1$ is
$(T\,T^\dagger)^{-1}$-pseudo hermitian. Viceversa, if $\Theta_1$ is
$(T\,T^\dagger)^{-1}$-pseudo hermitian for some operator $T$,
bounded with bounded inverse, and if the operator $T^{-1}\Theta_1 T$ admits an o.n. basis of $\Hil$ as eigenstates, then  $\Theta_1$ admits a basis of eigenvectors which is a Riesz basis,
with real eigenvalues.
\end{prop}
\begin{proof}
Let us first assume that $\Theta_1$ admits a basis of eigenvectors
$\F_1$ which is a Riesz basis, and that its eigenvalues are real:
$\Theta_1\varphi_n^{(1)}=\epsilon_n\varphi_n^{(1)}$, for all $n$.
Hence, as already stated, $\varphi_n^{(1)}=Te_n$ for a certain $T\in
B(\Hil)$, invertible with $T^{-1}\in B(\Hil)$, and an o.n. basis
$\E=\{e_n\}$. Hence the eigenvalue equation for $\Theta_1$ can be rewritten
as \be \Theta_{1,T}\,e_n=\epsilon_n e_n,\qquad n\geq 0,
\label{37}\en where $\Theta_{1,T}:=T^{-1}\Theta_1 T$. Of course this
means that $\Theta_{1,T}=\sum_{n\geq
0}\epsilon_n|e_n\left>\right<e_n|$ which, since $\epsilon_n$ is real
for all $n$, implies that $\Theta_{1,T}$ is self-adjoint. Now,
simple algebraic manipulations show that
$\Theta_{1,T}=\Theta_{1,T}^\dagger$ if and only if
$\Theta_1^\dagger=(T\,T^\dagger)^{-1}\Theta_1 (T\,T^\dagger)$, so
that $\Theta_1$ is $(T\,T^\dagger)^{-1}$-pseudo hermitian, \cite{mosta,mosta2,mosta3}.

Viceversa, let us assume that an operator $T$ exists, bounded with
inverse bounded, such that $\Theta_1$ is
$(T\,T^\dagger)^{-1}$-pseudo hermitian. Then
$\Theta_1^\dagger=(T\,T^\dagger)^{-1}\Theta_1 (T\,T^\dagger)$ which
implies that, defining as before $\Theta_{1,T}:=T^{-1}\Theta_1 T$,
$\Theta_{1,T}=\Theta_{1,T}^\dagger$. Hence, since an o.n. basis
$\E=\{e_n\}$ of eigenvectors of $\Theta_{1,T}$ exists by assumption,
$\Theta_{1,T}\,e_n=\epsilon_n e_n$, for all $n$, it follows that
$\epsilon_n\in\Bbb{R}$. It is further clear that, defining
$\varphi_n^{(1)}=Te_n$ and $\F_1=\{\varphi_n^{(1)},\,n\geq0\}$, this
set is a Riesz basis of eigenvectors of $\Theta_1$, with real
eigenvalues.

\end{proof}

A simple consequence of this Proposition is the following
\begin{cor}
Let us assume that, for a certain $T\in B(\Hil)$ with bounded
inverse, $\Theta_1$ is $(T\,T^\dagger)^{-1}$-pseudo hermitian. Let
us further assume that the IO $x$ is bounded and invertible and that
$[\Theta_1,N_1]=0$. Then $\Theta_2=x^{-1}\Theta_1 x$ admits a Riesz
basis of eigenvectors
$\F_2=\{\varphi_n^{(2)}=x^\dagger\varphi_n^{(1)}, \,n\geq 0\}$ with
real eigenvalues.
\end{cor}
The proof is straightforward and will not be given here. However, it
should be mentioned that $(\F_1,\F_2)$ are not biorthogonal in
general, and this makes the procedure non trivial. Indeed we have
already seen that, if $(\F_1,\F_2)$ are biorthogonal, then
$x=S_1^{-1}$ necessarily and, as a consequence, $\Theta_2=\Theta_1$.

\vspace{2mm}

{\bf Remark:--} It may be worth noticing that, as widely discussed
in \cite{bag5}, to any Riesz basis can be associated different
families of coherent states, or of generalized coherent states. In
particular, we can construct two {\em dual} families of coherent
states which, together, produce a decomposition of the identity.
This aspect of the theory will not be considered here.

\vspace{3mm}

We conclude this section reconsidering what we have done in
\cite{bag5} in the present language.

Our starting point is a Riesz basis of $\Hil$,
$\F:=\{\varphi_n,\,n\geq 0\}$. Then, if
$S=\sum_n|\varphi_n\left>\right<\varphi_n|$ is its frame operator,
we define $\hat\F=\{\hat\varphi_n:=S^{-1/2}\varphi_n,\,n\geq 0\}$.
This is an o.n. basis:
$\sum_n|\hat\varphi_n\left>\right<\hat\varphi_n|=\1$ and
$\left<\hat\varphi_n,\hat\varphi_k\right>=\delta_{n,k}$. Let now $A$
be a lowering operator defined by
$A\hat\varphi_n=\sqrt{n}\hat\varphi_{n-1}$, for all $n\geq 0$. In
particular this means that  $A\varphi_0=0$. Its adjoint is a raising
operator: $A^\dagger\hat\varphi_n=\sqrt{n+1}\hat\varphi_{n+1}$,
$n\geq0$, and they satisfy the canonical commutation relation
$[A,A^\dagger]=\1$. Let us now define
$$
a=S^{1/2}AS^{-1/2}\,\mbox{ and }\, b=S^{1/2}A^\dagger S^{-1/2}.
$$
Then $[a,b]=\1$ and, in general, $a\neq b^\dagger$. This is the commutation rule which defines the so-called pseudo-bosons, \cite{tri}. Moreover
$a\varphi_n=\sqrt{n}\varphi_{n-1}$ is a lowering operator for $\F$
while the related raising operator is $b$:
$b\varphi_n=\sqrt{n+1}\varphi_{n+1}$, \cite{bag5}. Now we put
$\F_1\equiv\F$, i.e. $\varphi_n^{(1)}=\varphi_n$ for all $n$, and $\F_2=\{\varphi_n^{(2)}=S^{-1}\varphi_n^{(1)}\}$,
and we see that $\varphi_n^{(1)}=\frac{1}{\sqrt{n!}}\,b^n\varphi_0$
and
$\varphi_n^{(2)}=\frac{1}{\sqrt{n!}}\,(a^\dagger)^n\varphi_0^{(2)}$,
where $\varphi_0^{(2)}=S^{-1}\varphi_0^{(1)}=S^{-1}\varphi_0$.
$(\F_1,\F_2)$ are biorthogonal and, defining $\Theta_1=ba$ and
$\Theta_2=\Theta_1^\dagger=a^\dagger b^\dagger$, we find that
$\Theta_1\varphi_n^{(1)}=n\varphi_n^{(1)}$ and
$\Theta_2\varphi_n^{(2)}=n\varphi_n^{(2)}$, for all $n\geq 0$.

Moreover, $S$ acts as an IO since $\Theta_1 S=S\Theta_2$. But, since
$S$ is invertible, this also implies that $\Theta_2=S^{-1}\Theta_1
S$ so that, recalling that $\Theta_2=\Theta_1^\dagger$, $\Theta_1$
is $S^{-1}$-pseudo hermitian.

The non-triviality of this example, i.e. the fact that $\Theta_2\neq
\Theta_1$, is based on the fact that the main assumption of
Proposition \ref{prop1} is violated:
$[\Theta_1,N_1]=[\Theta_1,S^{-2}]\neq 0$, in general. This can be
understood since we can show, first of all, that $[\Theta_1,S^{-2}]=
0$ if and only if $[A^\dagger A,S^2]=0$. But
$\left<\hat\varphi_l,[A^\dagger A,S^2]\hat\varphi_n
\right>=(l-n)\left<\hat\varphi_l,S^2\hat\varphi_n \right>$ which is
 zero, for all $l$ and $n$, if $S^2$ is diagonal in $\hat\F$
but not in general. Therefore, without further assumptions,
$[\Theta_1,N_1]\neq 0$.

Analogously, it is also possible to check directly that, but  if $S$
is diagonal in $\hat\F$, $\left<\hat\varphi_l,[A^\dagger
A,S]\hat\varphi_n \right>\neq 0$ and, as a consequence,
$[\Theta_1,S^{-1}]\neq 0$ and, yet, $\Theta_2\neq\Theta_1$.

\section{Conclusions}

We have considered some mathematical aspects of IO extending our
previous results also to the situation of non self-adjoint
operators. This has produced interesting results in connection with
PHQM. The role of Riesz
bases, in the present context, has been analyzed in some details and
they appear to be relevant alternatives to o.n. bases whenever we
look for eigenstates of a non self-adjoint operator.

We end the paper with a short summary of our point of view:

In order to relate the eigensystems of $\Theta_1$ and $\Theta_2$ it is
sufficient to have some intertwining relation $\Theta_2 x=x\Theta_1$
and it is not necessary that $x^{-1}$ exists. Indeed, if $\Phi$
is an eigenstate of $\Theta_1$ with eigenvalue $\epsilon$, then
$x\Phi$ is either zero or is an eigenstate of $\Theta_2$ with the
same eigenvalue.

However, if we want to talk of {\em standard} pseudo-hermiticity, $\Theta_2$ must
coincide with $\Theta_1^\dagger$ and $x$ must be invertible. But, since if
$x$ is not invertible our approach still works and produces
(quasi)-isospectral operators, we believe it may be worth investigating whether some other aspects of PHQM, other than the coincidence of the eigenvalues, can be extended in our more general settings. This work, which we have just began here, is now in progress.

\section*{Acknowledgements}

  The author acknowledges financial support by the Murst, within the  project {\em Problemi
Matematici Non Lineari di Propagazione e Stabilit\`a nei Modelli del
Continuo}, coordinated by Prof. T. Ruggeri.

\end{document}